\documentclass[sigconf]{acmart}

\usepackage{booktabs} 
\usepackage{balance} 

\setcopyright{acmcopyright}

\acmDOI{http://dx.doi.org/10.1145/3054977.3054995}

\acmISBN{978-1-4503-4966-6/17/04}

\acmConference[IPSN 2017]{The 16th ACM/IEEE International Conference on Information Processing in Sensor Networks}            
{April 2017}{Pittsburgh, PA USA}

\begin{document}
\title[ADN]{ADN:  An  Information-Centric Networking Architecture  for the Internet of Things}

\author{J.J. Garcia-Luna-Aceves}
\affiliation{%
  \institution{University of California at Santa Cruz}
  \streetaddress{Computer Engineering Department}
  \city{Santa Cruz} 
  \state{CA} 
  \postcode{95064}
}
 \email{jj@soe.ucsc.edu }

\begin{abstract}
Forwarding data by name has been assumed   to be  a necessary aspect of an information-centric redesign of  the current Internet architecture  that makes content access, dissemination, and storage more efficient.    The Named Data Networking (NDN) and  Content-Centric Networking (CCNx) architectures are the leading examples of such an approach. However,  
forwarding data by name incurs storage and communication complexities that are orders of magnitude  larger than solutions based on forwarding data 
using addresses. Furthermore, the specific algorithms used in NDN and CCNx have been shown to have a number of  limitations. 
The  {\em Addressable Data Networking} (ADN) architecture is introduced as  an alternative  to  NDN and CCNx. ADN  is particularly attractive for large-scale deployments of  the Internet of Things (IoT), because it requires far less storage and processing in relaying nodes than NDN. 
ADN  allows things and data to be denoted by names, just like NDN and CCNx do. However, instead of replacing the waist of the Internet with named-data forwarding, ADN uses an address-based forwarding plane and introduces an information plane that 
seamlessly maps names to addresses without the involvement of end-user applications.
Simulation results illustrate the order of magnitude savings in complexity that can be attained with ADN compared to NDN.

\end{abstract}

%
%





\maketitle

\section{Introduction}

Information-Centric Networking (ICN) architectures  \cite{icn-survey1, 
icn-survey3, icn-survey2}   have been proposed to improve the quality of information perceived by consumers compared to the current IP Internet. 
The Named Data Networking (NDN)  \cite{ndn} and the Content-Centric Networking (CCNx)  \cite{ccnx} architectures advocate the use of what has been called a ``stateful forwarding plane" \cite{ndn-fw2}.  In this approach,  consumers request content by name issuing Interests that state the names of the required content objects (CO). Routers maintain Forwarding Information Bases (FIB) listing the next hops to name prefixes and  use that information to forward Interests towards the producers of content.  The requested COs or feedback are sent back to consumers by means of Pending Interest Tables (PIT) that list the interfaces over which responses should be sent back for each Interest that has been forwarded by a router.

It is well documented by now  that  using FIBs to maintain routing state for each name prefix and using PITs to maintain per-Interest forwarding state 
requires far more storage than the traditional address-based routing and forwarding approach used in the Internet today 
\cite{dai-12,  peri-11, song15, var-13, vir-13, wahl13b}, considerable research has focused on trying to make stateful forwarding planes (i.e., having forwarding state for each Interest) more efficient  \cite{afan15, papa-14, panini16, so12, so13, song15, tsi14, var-12, var-13, wang12, wang13, yuan-12}.
The proponents of NDN and CCNx  have argued that the additional processing and storage  costs 
incurred by named-data forwarding are justified by the benefits derived from  it, which consist of:
(a) enabling adaptive forwarding disciplines that provide for loop-free multipath data retrieval and native support of multicast; (b)
providing efficient recovery from packet losses;  (c) allowing efficient flow balancing and congestion control; and (d) robustly detecting and recovering from forwarding problems.   

On the other hand,  we have proposed a few alternatives  that replace PITs with much smaller  tables  for  the forwarding of responses to Interests \cite{ifip2016, globe2016},  or replace both PITs and name-prefix FIBs with smaller  tables for the forwarding of both Interests and responses to them \cite{icn2016}.   
These prior results indicate that the performance benefits  ascribed to the use of a stateful forwarding plane in NDN and CCNx are not the result of using names instead of addresses for the forwarding of Interests and data. 
Motivated by these results and the need to make efficient use of storage and processing resources in IoT devices,  we introduce the   {\em Addressable Data Networking} (ADN) architecture. 
ADN  is an 
information-centric alternative to named-data forwarding that requires far less complexity in the forwarding plane.

Section \ref{sec-prelim}  summarizes of the operation of the forwarding plane  of NDN to provide the necessary context for the description of ADN. 

Section \ref{sec-ip} discusses the challenges that IoT poses to both the IP Internet and the NDN and CCNx architectures, and motivates the need for a fresh look at how information is discovered and switched in an IoT.

Section \ref{sec-adn}  introduces   the
{\em Addressable  Data Networking} (ADN) architecture in which Interests are forwarded based on the addresses of data and responses are sent back using 
addresses of sources of Interests. ADN  can be instantiated in a number of ways, including some of the other ICN architectures proposed to date \cite{icn-survey3,  ifip2016, icn2016} or modifications to the algorithms used in IPv4 and more importantly IPv6  today. 
Finding out what  the most efficient instantiation of such a forwarding plane would be  is beyond the scope of this paper and deserves further study.

Section \ref{sec-loop} 
shows that remembering the names and nonces for each forwarded Interest cannot ensure that every Interest is  consumed by either a  data packet or a NACK in the absence of packet losses, because  multiple Interests may be aggregated in PITs along forwarding loops. This  results in a new type of forwarding-deadlock problem in which Interests  ``wait to infinity" for responses that never come and can be discarded only after their PIT lifetimes expire. In turn, this can prevent  native support for synchronous or asynchronous multicast from working correctly, because Interests may fail to establish multicast forwarding trees for data.

Section \ref{sec-perf} presents the results of simulation experiments aimed at highlighting the benefits of using addresses rather than names for the forwarding of Interests.
Section \ref{sec-conc} discusses how ADN can implement fast recovery from packet losses and 
congestion-control mechanisms that are at least as responsive as those in NDN but without the need for per-Interest forwarding state, and summarizes
our final thoughts on  research directions for ADN.

\section{The NDN Forwarding Plane }
\label{sec-prelim}

Routers in NDN use Interests, data packets, and negative acknowledgments (NACK)  to exchange content \cite{ndn-fw2, ndn-paper}. 
An Interest is identified in NDN by the name $o$ of the CO being requested and a nonce created by the origin of the Interest. A data packet includes the CO name, a security payload, and the payload itself. A NACK carries the information needed to denote an Interest and a code stating the reason for the response. 
To process Interests, data packets, and NACKs, each router uses a content store (CS), 
a forwarding information base (FIB),  and a pending Interest table (PIT).  Fig. \ref{ndn-fig} illustrates the forwarding approach in NDN.

A CS is a cache for COs stored locally and indexed by their names.  
The FIB entry  for  a given name prefix  lists the stale time for the entry, and a list of interfaces ranked according to a forwarding policy.
The information for an interface includes a routing preference, a round-trip time (RTT), a status, and a rate limit \cite{ndn-fw2}. FIBs are populated using a content routing protocol  and  a router matches Interest names stating a specific CO   to FIB entries 
corresponding to prefix names using \emph{longest prefix match}. 
The entry in  the PIT  of a router for a given CO 
consists of  one or multiple tuples stating a nonce received in an Interest for the CO,  the incoming interface where it was received and a lifetime, and a list of the outgoing interfaces over which the Interest was forwarded and a send time.

\begin{figure}[h]
\includegraphics[height=2.1in, width=3.1in]{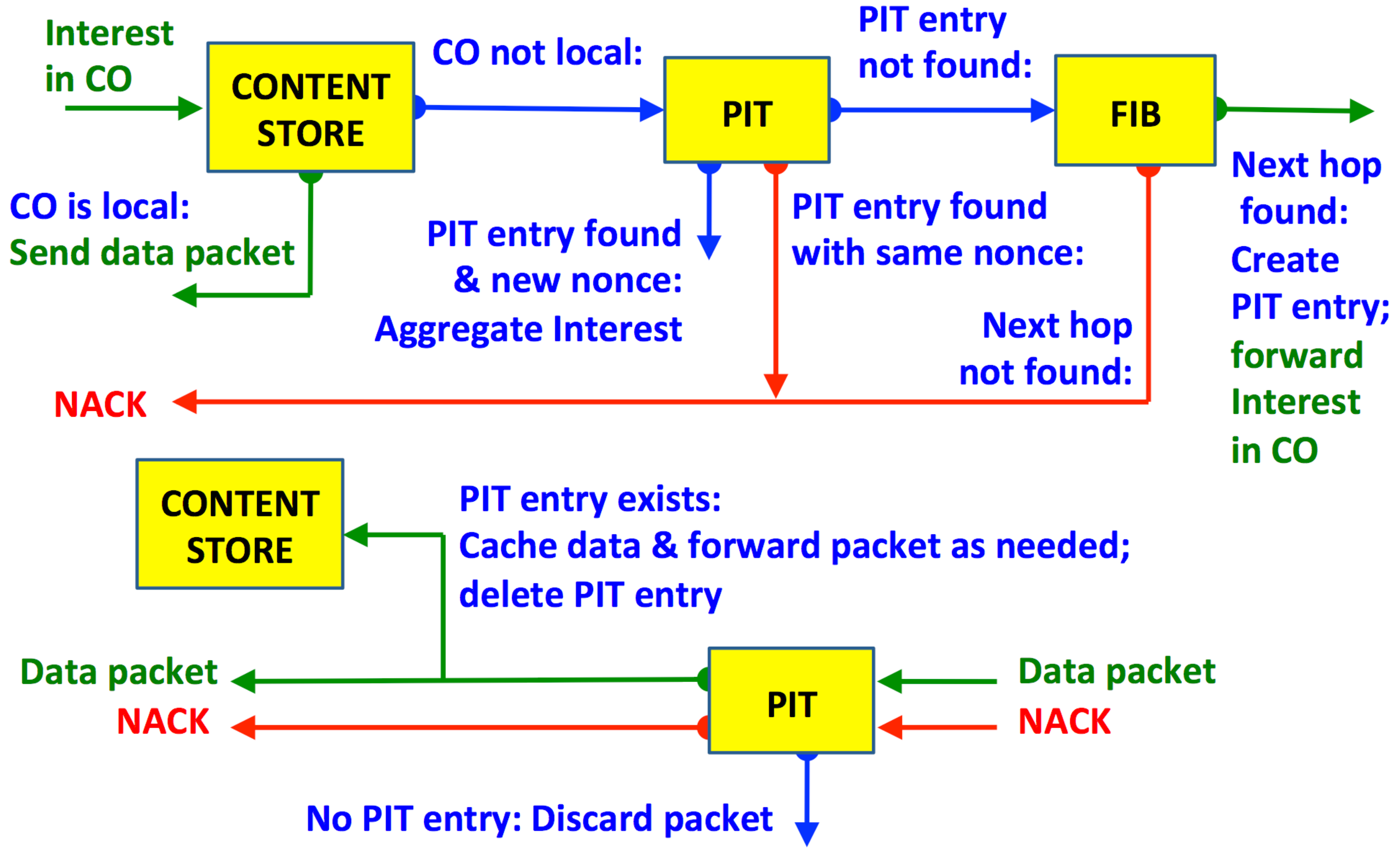}
\caption{Forwarding in NDN}
   \label{ndn-fig}
\end{figure}

When a router receives an Interest, it checks whether there is a match in its CS for the CO requested in the Interest. The Interest matching mechanisms used can vary, and to simplify the comparison between forwarding planes we can assume that exact Interest matching is used.  If a match to the Interest is found, the router sends back a data packet over the reverse path traversed by the Interest. If no match is found in the CS, the router determines whether the PIT stores an entry for the same content.   If  the Interest states a nonce that differs from those stored in the  PIT entry for the requested content, then the router ``aggregates" the Interest by adding the incoming interface from which the Interest was received and the nonce to the PIT entry without forwarding the Interest.  If the same nonce in the  Interest is already listed in the PIT entry for the requested CO, the router sends a NACK over the reverse path traversed by the Interest.
If a router does not find a match in its CS and PIT, the router forwards the Interest along a route (or routes) listed  in its FIB for the best prefix match.     
A NACK can be sent by a content provider if no match is found for the CO name stated in the Interest.

Based on the information in its FIB and PIT, each  router establishes an interface ranking \cite{ndn-probing}  to determine the interfaces that  should be used to forward Interests in order use the best paths over which content should be retrieved.  The ranking of interfaces classifies interfaces into classes based on perceived performance, and uses periodic measurement of performance and probing.

\section{Limitations of Approaches based on IP or NDN}
\label{sec-ip}

An IoT deployment involves large numbers of resource-constrained devices designed for low manufacturing cots and limited operational expenses. This results in  IoT devices that  are power constrained and have limited processing, storage and communication functionalities compared to routers and end systems connected to wired segments of the Internet. The following paragraphs point out a number of problems involved in using either the IP Internet stack or the NDN and CCNx architectures as they are currently defined to enable large-scale IoT deployments. We do not address the physical and link layers directly, because they must be designed to enable  small and cost effective IoT devices.

\subsection{Network-Layer IoT Challenges}

At the network layer, the IoT challenges stem from the mismatch between the physical  characteristics of devices and the requirements in the protocol stacks of  IP and NDN or CCNx to denote destinations and forward data to and from destinations.

Due to energy constraints of IoT devices,  the maximum size of packets at the link layer in IoT deployments needs to be very small. On the other hand, IoT applications are many and will just increase in numbers in the future, which means that naming of data must be expressive to satisfy their varied requirements. This mismatch is not easily solved with the IP or  NDN/CCNx architectures. 

IPv6 requires networks to support a minimum 
of 128-byte maximum transmission unit (MTU) to avoid packet fragmentation, and NDN uses application-friendly expressive names as an integral part of data forwarding functionality.
For IPv6, the mismatch between link-level MTU sizes and IPv6 introduces unnecessary adaptation-layer complexities. For NDN and CCNx, the use of large names for data forwarding 
incurs unnecessary processing and storage requirements in IoT relays, and 
introduces the need for fragmentation and its associated problems.

The limitations of the IP Internet and NDN or CCNx architectures go beyond the mismatch between  MTU lengths at layer two and packet headers and the length of names used to denote data. 
The forwarding mechanisms used in these architectures are not well suited for large-scale IoT deployments. 

Even though routing and forwarding in an IP network is based on fixed-length identifiers, routers maintain routes to all possible destinations  proactively and the forwarding plane has no means to reflect the ordering among routers that is inherent in the computation of routes to destinations. As a result, packets may traverse undetected forwarding loops and the best routers can do in such a case is drop packets after they traverse too many hops. Many of the proposals that have been advanced for routing in the context of IoT deployments rely on the same protocols designed for mobile ad hoc networks (MANET), or use a spanning tree of the network in which the routers near the root of the tree must maintain routes to most devices (e.g., RPL \cite{rpl}) and use source routing to avoid having to maintain large routing tables at some nodes. The end result is routing tables that become too large with entries for each host and packet headers that  can grow too large with source routes.

Another limitation of IP routing when applied to IoT deployments is the way in which multicasting is supported. A multicast routing protocol is needed to establish routing state for multicast groups, network-level addresses must be mapped to link-level addresses, and multicast sources end up transmuting at a rate that is acceptable to the slowest receiver. The latter is a consequence of the network not providing any in-network storage. The few alternatives to IP multicast that have been proposed for low-power and lossy networks amount to flooding, which is not acceptable in large-scale deployments.

On the other hand, routing in NDN and CCNx  eliminates the need to maintain routes for each host in the network, and  no additional multicast routing protocol is needed to establish and maintain multicast forwarding state.
However, some of the routing protocols proposed for NDN require routers to maintain information about the network topology and each name-prefix replica in the network, and more importantly  each router must maintain forwarding state for every Interest they forward, which may become onerous on some relays.
Furthermore, although Interests are prevented from traversing forwarding loops multiple times in NDN, forwarding deadlocks may be created in either NDN or CCNx that are just as harmful, as we discuss in Section \ref{sec-loop}.

\subsection{Transport-Layer IoT Challenges}

The connection-oriented approach to reliable data transfer and congestion control used in the IP Internet is not a good match for IoT deployments, because of the characteristics of IoT devices and the traffic induced by IoT applications. Many IoT devices may require to use on-off cycles to extend their  battery lives, and a considerable amount of IoT applications involve short transactions for which establishing connections  simply induced unnecessary latencies. In addition, the interaction between TCP  and link-level retransmissions or losses due to physical-layer effects of wireless links can make the performance of TCP suffer \cite{tulip}.

An approach to avoid the problems with TCP operating over an IoT network consists of using UDP as the transport protocol and implementing the retransmission and congestion-control functionalities in application libraries. Unfortunately, this makes application developers responsible for the design and implementation of functionality that should be transparent to applications and limits the ability of such applications to evolve in parallel with IoT technologies.

The NDN and CCNx architectures split the transport-level functionality between application libraries and the forwarding mechanism. Consumer applications are responsible for pacing the sources of data by the rates at which they submit Interest in data, and routers maintain per-Interest forwarding state to allow for error recovery and retransmissions. The limitations with this approach are that it incurs considerable forwarding-state  overhead in IoT routers, and just like servers can be the subject of  SYN flooding attacks designed to exploit the state kept at servers for the connection-establishment phase of TCP, malicious users can simply inject Interests aimed at overwhelming the forwarding state of relays.

\subsection{Application-Layer IoT Challenges}

Resource discovery and enabling opportunistic in-network caching are  essential components of the information plane of an IoT deployment. They allow applications to denote resources and services by name, and content to be delivered efficiently to resource-constrained devices through relays that also have resource constraints.

Resource discovery could be attained in the context of the protocol stack of the IP Internet using the DNS or augmenting  the routing protocol used in the network. However, both approaches have limitations.
Using the domain name system (DNS) for resource discovery in an IoT would require extensions to the DNS service discovery \cite{dns-ds} that include mechanisms for IoT devices and hosts to determine how to contact the directory servers in the network. Multicast DNS \cite{mdns} can be used to avoid having to configure hosts with the addresses of directory servers, but is applicable only in very small IoT deployments in which energy constraints of relaying nodes is not an issue.

Augmenting the routing protocol running in an IoT to support resource discovery functionality  is similar to the content routing protocols that have been proposed for NDN, in that the signaling of the routing protocol is used to disseminate information about network resources. A  virtual link can be assumed between a resource and the node hosting the resource, and link state updates can be used to disseminate the presence of the resource and obtain routes to it. 
This approach is not suitable for large-scale IoT deployments.  The signaling overhead incurred in disseminating resource information through the routing protocol is much larger than the overhead of traditional routing, which is not applicable for resource-constrained IoT relays. 

In-network caching of content is a major benefit of NDN and CCNx, as it enables content delivery from caching sites near the consumers of content. Although transparent caching is available in the protocol stack of the IP Internet, it requires caching at the edge by servers that intercept all requests, and assumes that resource discovery takes place (e.g., determining the address of the content provider in case of a cache miss). Whether the protocol stack of the IP Internet, NDN, or CCNx are used, a big challenge to address in the future is the ability to cache content opportunistically without the caching sites being able to read the cached content.

 \subsection{Implications on The Network Architecture}

From the previous discussion  we  argue that a large-scale  IoT deployment requires the introduction of an information plane and a switching plane that work in coordination but use separate mechanisms.  A switching plane can  allow the use of small MTUs friendly to the link levels of an IoT and data forwarding based on small MTUs, while an information plane can support expressive names needed by IoT applications.
Separating the information plane from the switching plane enables the use of simple routing and  forwarding mechanisms that can be implemented in resource-challenged devices, and enables the use of  resource discovery, in-network caching, and congestion control  mechanisms at edge devices and  relays in ways that do not increase the complexity of the switching plane (i.e., the forwarding state required in IoT relays, and the signaling overhead incurred in updating forwarding state).

\section{ADN: A State-Light Forwarding Plane}
\label{sec-adn}

Addressable Data Networking (ADN) is based on the realization that the content-delivery limitations associated with the  IP Internet are not due to the use of addresses in datagram forwarding,  but rather from  the limitations in the forwarding and routing mechanisms being used and the way in which  names are mapped into addresses.

ADN  uses  Interests, data packets, and NACKs  to exchange content just like NDN does. The key difference between  ADN and NDN is that addresses rather than names are used to  forward Interests and send responses. 
The approach we describe is based on techniques we have introduced recently \cite{ifip2016, icn2016} and is simply intended to show that: (a) maintaining per-Interest forwarding state is not needed, and (b)  using addresses for data forwarding  is in fact more efficient  than using CO names. ADN could be instantiated in the context of other ICN architectures or even IP.

We use the term {\em anchor} to denote a router that, as part of the operation of  either a name-based routing protocol or a content-directory protocol, announces  all the content corresponding to a name prefix being locally available. 
If multiple mirroring sites host the content corresponding to a name prefix, then the routers attached to those sites announce the same name prefix. 
A   router  caching COs from a name prefix does not announce the name prefix in the name-based routing protocol or content-directory protocol.  Hence, an anchor of a name prefix is also an anchor of each CO corresponding to the name prefix.

Fig.~\ref{adn-fig} illustrates the forwarding approach adopted in ADN for the case of edge caching, i.e., 
content caching is done  only at those nodes connected directly  to consumers requesting content.
In this case, the information plane resides at edge systems, which are either servers or routers connected directly to content consumers. The switching plane resides in all routers and switches.

A router uses  a FIB and a {\em  source address table} (SAT) that replaces the  PIT. In addition, a router may also maintain a content store (CS) to cache content, and a {\em directory information base} (DIB) to assist with the binding of name prefixes to addresses.  Routers that receive Interests from consumers or send responses to consumers implement all the functions illustrated in the figure, while routers that simply forward Interests and responses from other routers (i.e., relays) need to implement only  those functions shown in the blue-shaded areas.   The functionality corresponding to the information plane can be implemented at edge routers or edge servers, depending on the specific IoT deployment, while the simpler switching-plane functionality is implemented at each router.

\begin{figure}
\includegraphics[height=2.5in, width=3.5in]{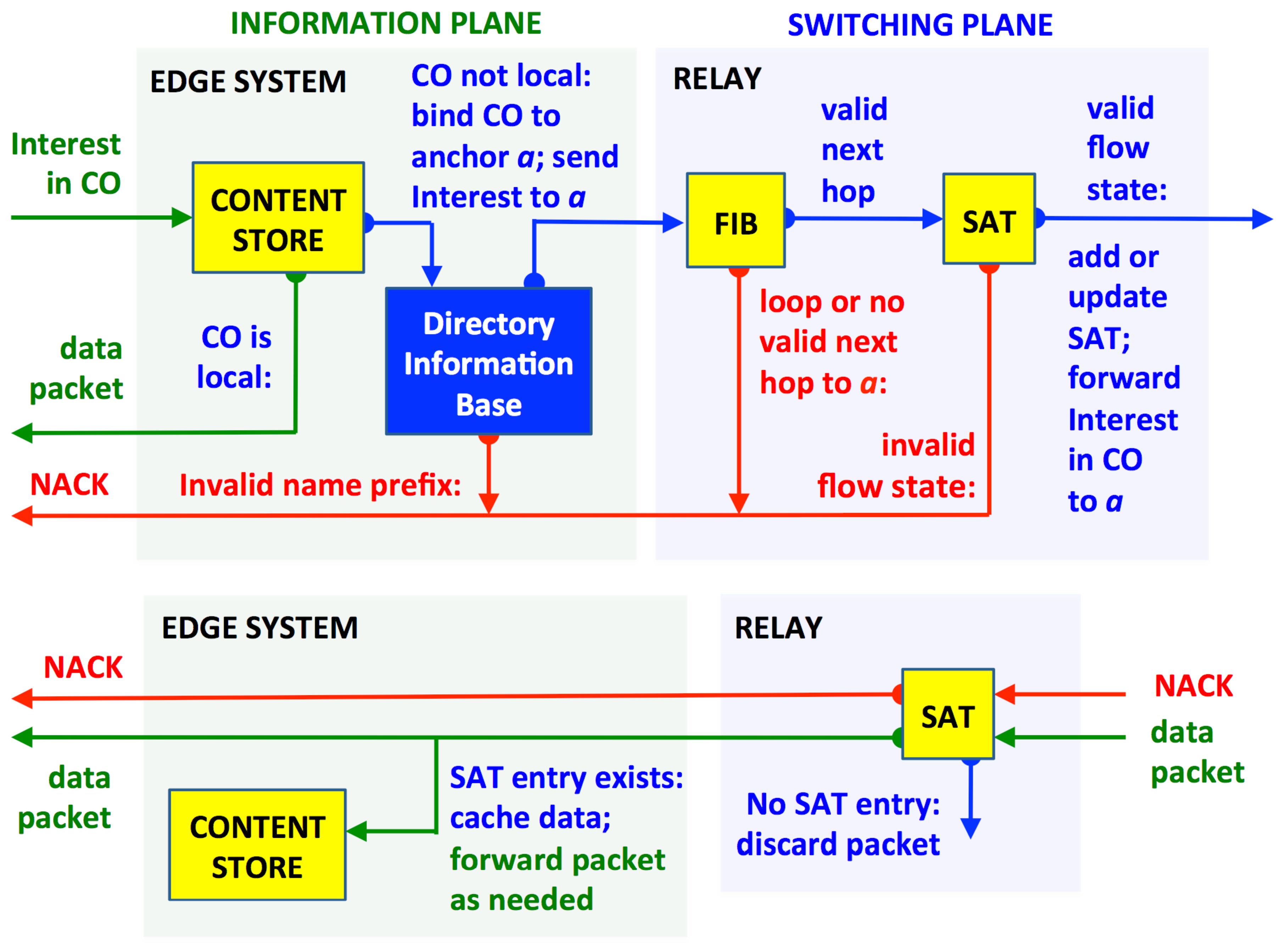}
\caption{Forwarding in ADN}
   \label{adn-fig}
\end{figure}

The CS maintained by a router in ADN is the same as in NDN. However, 
the FIB maintained by a router does not list forwarding state for name prefixes, it  lists the next hops {\em and} distances through such interfaces to  anchors of name prefixes.  The information regarding which anchors store the content corresponding to  different name prefixes is stored  in the DIBs.

Consumers request  COs by name in the form of Interests that simply state CO names.
It is assumed that a naming convention is in place to allow routers to determine whether an Interest received from a consumer corresponds to a unicast or multicast flow from the name of the CO being requested. 

A router processing an Interest in a CO from a consumer determines whether  the CO being requested is cached in its CS, and if so resolves the request. If the CO is remote, the router must bind the CO name to the address of an anchor of a name prefix that is the best match for the CO name.
This is attained  using the information stored in the DIB, and the result is
that any Interest forwarded by a router to another router includes the address of an anchor for the CO being requested.

The binding of CO names to anchor addresses 
can be implemented in  a number of ways. 
We have recently proposed an approach \cite{icn2016} in which a name-based content routing protocol just like the one used in NDN or CCNx (e.g., NLSR \cite{nlsr} or DCR \cite{dcr}) populates  the DIB and the FIB at each router.  This approach 
disseminates all the name-prefix to anchor mappings to all routers and 
incurs the same overhead populating DIBs  and FIBs as NDN requires to populate FIBs listing name prefixes. However, it is also possible to use an address-based routing protocol to populate the FIBs with entries for addresses or address prefixes corresponding to anchors, and a separate directory update protocol to disseminate the mappings between name  prefixes and anchor addresses. 

The  {\em  source address table} (SAT) replaces the PIT and  stores an entry for each  Interest source address  for which Interests have been forwarded by the router, rather than forwarding state for each Interest.  This table is needed in the ADN architecture if routing information is maintained proactively only for anchor addresses, rather than for all network nodes as it is done in the IP Internet.

Prior results  \cite{caching}  indicate that most of the benefits derived from in-network caching can be attained just with edge caching;
however, on-path caching may still provide performance improvements over edge caching depending on the nature of the content.  Fig.~\ref{adn-cache-fig} illustrates that relay nodes in ADN can also use on-path caching, even if Interests state a target anchor. In this case, the information plane is instantiated in relay nodes as well as  in edge systems. A relay node receiving an Interest processes it by first comparing the CO name in the Interest to the CO names stored in its CS, and the Interest is forwarded to the anchor stated in the Interest if the CO is not stored locally, provided that there is a route to the intended  anchor and no forwarding loops exist.

\begin{figure}
\includegraphics[height=1.5in, width=3in]{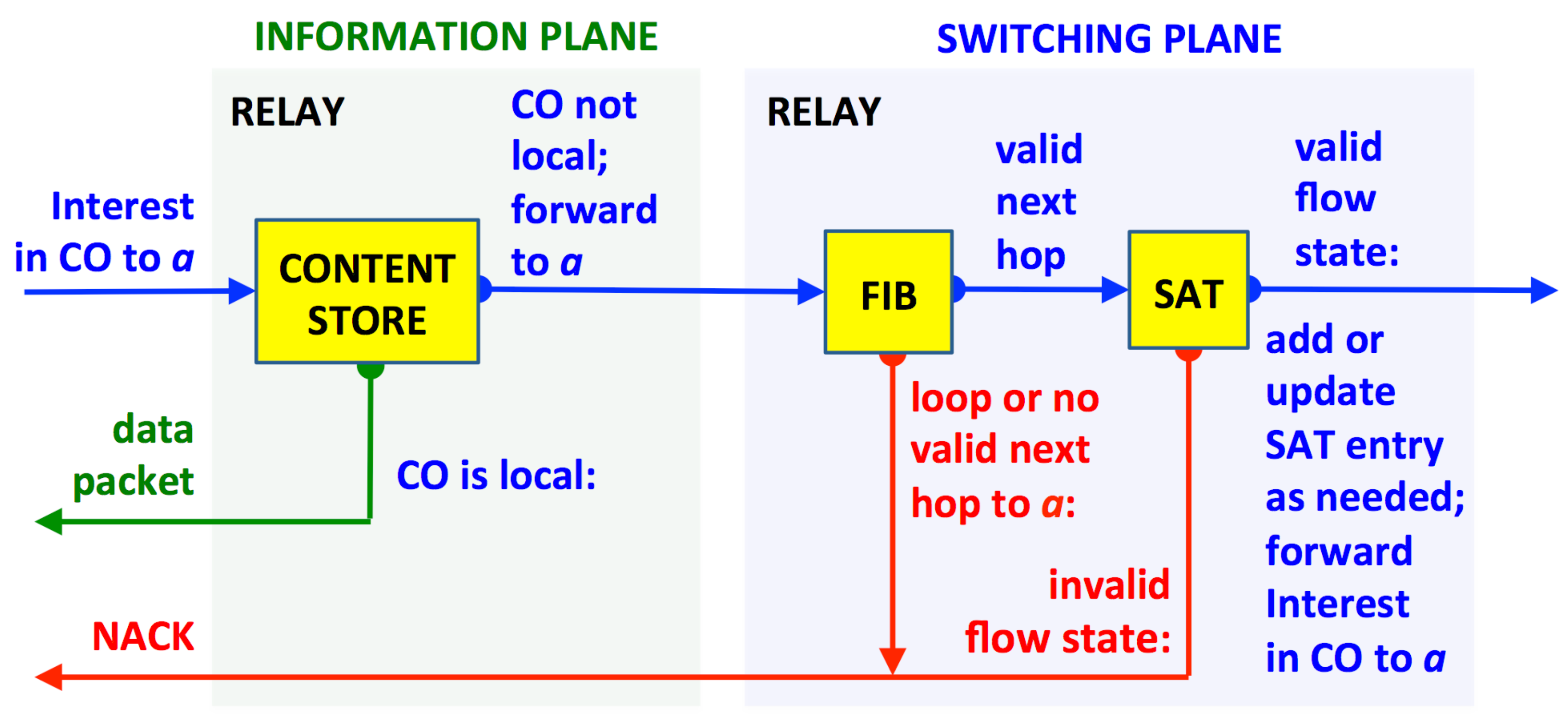}
\caption{Forwarding in ADN using on-path caching}
   \label{adn-cache-fig}
\end{figure}

In ADN, Interests that are part of  a unicast flow are forwarded with no aggregation, and caching of COs is the only mechanism used to suppress the forwarding of multiple Interests  from different consumers requesting the same CO. This design  decision  is based on the results of our recent analysis of Interest aggregation in NDN \cite{ali-ifip16}.
By contrast, Interests in a multicast flow need to be aggregated before COs flow back toward consumers. This is the case because Interests can then be used to build the multicast forwarding tree without the need for additional signaling in the control plane. In ADN, the name of the CO being requested gives an indication that the request corresponds to either a unicast or multicast flow, and a 
flow-state value is used to enable Interest aggregation for multicast flows in the SAT entries maintained by routers.

The distance to  anchor address $a$ 
through the interface to neighbor $k$ listed in  $FIB_i$ (the FIB of router $i$)  is denoted by  $h_i(a, k)$.  The SAT entry at router $i$ for Interest source address $s$ is denoted by $SAT_i (s)$ and states:  the address $s$  denoting either the source of a unicast Interest or a multicast group, a flow-state value ($f_i(s)$), and a set $O_i(s)$ of  one or more interfaces  towards the origin of Interests. 

Data packets and NACKs are similar to those used in NDN; however, a response to an Interest sent by router $k$ carries additional forwarding information, namely: an address denoting the source of the Interest that originated the response ($s^R(k)$), and  a  flow-state value ($f^R(k)$).

As we have noted, an Interest sent by a consumer $c$ simply states the name  $o$  of a CO, and is denoted by $I_c[o]$.  By contrast, an Interest sent by  router $k$ requesting a CO with name $o$  is denoted by $I[o, s^I(k), a^I(k), h^I(k), f^I(k)]$ and  
states the following:  The name of the requested CO ($o$), an address denoting the source of the Interest ($s^I(k)$),  the address of  an intended anchor of the CO ($a^I(k)$), a distance to the anchor ($h^I(k)$), and a flow-state value ($f^I(k)$).
The flow-state value carried in an Interest determines whether the Interest  is part of a unicast or multicast flow.  The value $f^I(k)= 0$  denotes a unicast flow, and  $f^I(k) >  0$  denotes a multicast flow.

As can be observed from Figs. \ref{ndn-fig} and \ref{adn-fig}, forwarding of responses to Interests in NDN and ADN are similar. Responses simply traverse  reverse paths using the traces stored in PITs in NDN and SATs in ADN.  However,  Interest forwarding in ADN is very different than in NDN. 

Let $S_i(a)$ denote the set of next hops to the address range that is the best match for anchor  address $a$  in $FIB_i$. 
Router $i$  processes  $I[o, s^I(k), a^I(k), h^I(k), f^I(k)]$ 
according to  the following rules.

\vspace{0.1in}
{\fontsize{9}{9}\selectfont
\noindent
{\bf Interest Forwarding Rule (IFR) at router $i$:}  \\
Forward   $I[o, $ $s^I(k), $ $a^I(k), $ $h^I(k), f^I(k)]$ if \\
$ \exists ~v  ~(~ v \in  S_i ( a^I(k) )  ~\wedge~ h^I(k)  > h_i(a^I(k), v) ~) ~ \wedge$
 \\
$[ ~(~f^I(k)= 0 ~) ~~\vee $
$( \not\exists ~SAT_i( s^I(k) ) ~ \wedge ~f^I(k) = 1 ~)  ~~\vee$  
\\
$~~~(~\exists ~SAT_i( s^I(k) ) ~ \wedge  ~f^I(k) = f_i( s^I(k) ) + 1 ~) ~~] $  
}

 \vspace{0.05in}
{\fontsize{9}{9}\selectfont
\noindent
{\bf Interest Aggregation Rule (IAR) at router $i$:}  \\
Aggregate  $I[o, $ $s^I(k), $ $a^I(k), $ $h^I(k), f^I(k)]$ if \\
$ \exists ~v  ~(~ v \in  S_i ( a^I(k) )  \wedge h^I(k)  > h_i( a^I(k) , v) ~)  ~ \wedge$
 \\
$[~~\exists ~SAT_i( s^I(k) ) ~ \wedge  ~f^I(k) = f_i( s^I(k) )  ~) ~~] $  
}

\vspace{0.05in}
{\fontsize{9}{9}\selectfont
\noindent
{\bf Interest Negation  Rule (INR) at router $i$:}  \\
Send a NACK to  $I[o, $ $s^I(k), $ $a^I(k), $ $h^I(k), f^I(k)]$ if \\
$ \not\exists ~v  ~(~ v \in  S_i (a^I(k) )  ~\wedge~ h^I(k)  > h_i( a^I(k), v) ~) ~ \vee$
 \\
$[~\exists ~v  ~(~ v \in  S_i (a^I(k))  \wedge h^I(k)  > h_i(a^I(k), v) ~)  ~ \wedge$
\\
$(~\exists ~SAT_i( s^I(k) ) ~ \wedge  ~f^I(k) \not= f_i( s^I(k) ) + 1 ~) ~~] $  
}
 
 \vspace{0.1in}
These three rules ensure that router $i$  forwards or aggregates an Interest only if it can find a valid next hop and the Interest has a valid flow state, and sends a NACK back otherwise (see Fig. 2).
A valid next hop for a router  is a neighbor  through which it has a distance to the intended anchor that is smaller than the distance stated in the Interest.  

The name of any CO  is assumed to contain information  that denotes whether the CO corresponds to a multicast service or not.  This informs an ingress router whether to forward a unicast or multicast Interest when it processes an Interest from a local consumer. 
A unicast Interest carries a a flow-state value of 0 and a source address that can have global scope as in the IP Internet today, or only local meaning as we have discussed in recent  content-centric networking approaches \cite{icnc16, ifip2016, icn2016}.
A multicast Interest carries a flow-state value of at least 1 and a source address that denotes a multicast group. The first Interest in a multicast flow
states a flow-state value of 1. 

A multicast Interest  is forwarded  if it carries a flow-state value indicating the next CO expected from the multicast source, and it is aggregated if it carries a flow state of equal value to the one stored in the SAT entry for the multicast group. Different   source-pacing mechanisms are possible using flow-state information, such as those presented in \cite{globe2016, icnc2017}.

The address of the multicast group  $a(g)$ is simply an identifier used to denote the receivers of the group and forward responses back to them. In practice, $a(g)$ can obtained directly from the name $g$ (e.g., by having the address be part of the group name or by defining a has function of the name).

\section{Multipath Data Retrieval  
\\ and Native Support for Multicast }
\label{sec-loop}

\subsection{Deadlocks in NDN Unicast and Multicast Forwarding}

According to NDN, a router determines that an Interest has traversed a loop when it receives an Interest for which an entry exists in its PIT stating the same name and nonce carried in the Interest, and a router simply adds an incoming interface for an Interest stating a CO name if the nonce is different than those it has stored in its PIT for the same CO name.   

On the other hand, a name-based multipath routing protocol is used to populate the FIBs, and each router ranks the interfaces in a FIB entry into three classes \cite{ndn-fw2}: (a) {\em green}, which means that the interface can bring data back; (b) {\em yellow}, which states that it is unknown whether the interface may bring data back; and (c) {\em red}, which denotes the inability of brining data back through that  interface. 
This  ranking  is done independently of the routing protocol, which is only required to build routes based on long-term path characteristics.   Green interfaces are always preferred over yellow interfaces, and a  green interface turns yellow when a pending Interest times out, data stops flowing for some time, or a NACK stating  	``No Data" or ``Duplicate" is received. An interface is marked red when it goes down.

Based on the premise that Interests and hence data packets cannot loop in NDN, it has been assumed that  \cite{ndn-fw2}: 
(a) routers may try multiple alternative paths in Interest forwarding; (b) 
the routing protocol only needs to disseminate long-term changes in topology and policy, without having to address near-term changes; and (c) 
multicast trees are formed for data to return by the fact that  Interests from different sources requesting the same content are suppressed and  only the first Interests are forwarded towards the producers of content. 
Unfortunately, although the same Interest cannot traverse a forwarding loop and data packets cannot traverse loops,   NDN  is subject to 
{\em forwarding deadlocks} that are just as harmful. 

Fig.~\ref{ndn-loop-fig} shows a simple example of this  problem  using a small eight-router network in which  three consumers  ($a$, $b$, and $c$) request multicast content from a source $S$ advertising content for a multicast group  $g$.   
The colored directed links illustrate the interface ranking at each router for the FIB entry
for name prefix $g^*$.
Routers $x$ and $v$ have ranked their interfaces to $u$ as yellow, and router $u$ has ranked its interface to $w$ as yellow as well. 
Interests are forwarded to CO names $o_i$, which are part of a name prefix $g$ used to denote all the content for  multicast group $g$. In turn, the name $g$ is part of a name prefix $g^*$ for which routers have FIB entries. 

The interfaces from $x$ to $k$ and from $v$ to $m$ go down at times $t_1$ and $t_2$, respectively, and the three consumers of the multicast group send Interests requesting more content around the same instant $t_3$. A dashed line indicates an Interest being submitted or forwarded, the time when it is sent, and its origin.

\begin{figure}
\includegraphics[height=1.5in, width=3in]{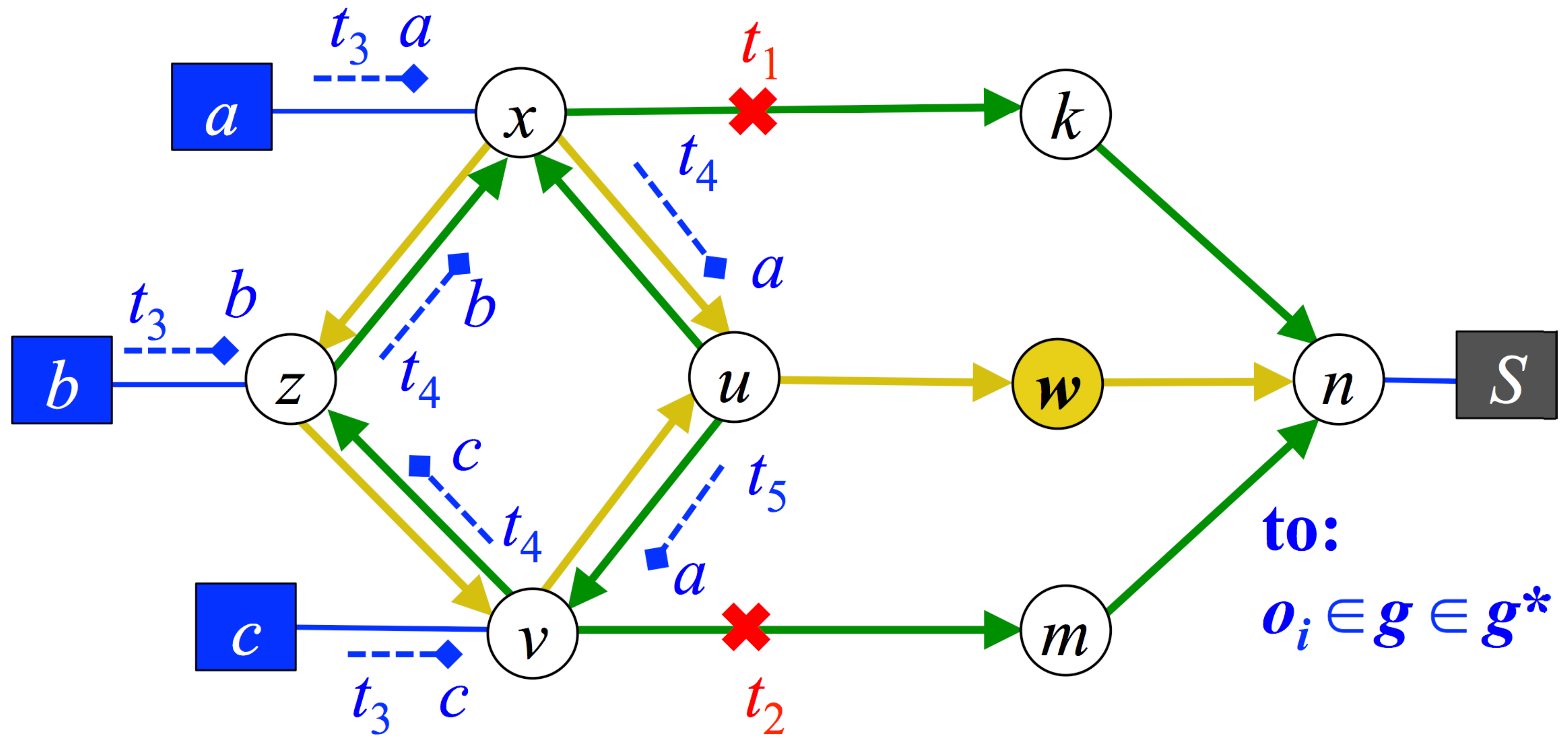}
\caption{Forwarding deadlock  in NDN}
   \label{ndn-loop-fig}
\end{figure}

When router $x$ receives the Interest from consumer $a$, it can use any of its yellow interfaces for $g^*$, which is the best match for any CO name in multicast group $g$. Accordingly, it  forwards  the Interest to router  $u$ at time $t_4$.  
Concurrently, routers $v$ and $z$ forward the Interests they receive using the green interfaces they have for the FIB entry that is the best match for the CO name. Router  $u$ forwards the Interest originated at $a$ to router $v$ at time $t_5$, because it has a  green interface for name prefix $g^*$.

Router $x$ aggregates the Interest   received from router $z$ after time $t_4$, because it has a pending Interest for the same CO name and a different nonce in its PIT. Similarly, router $z$ aggregates the Interest received from router $v$ after time $t_4$, and router $v$ aggregates the Interest received from router $u$ after time $t_5$. It is clear that, although no one Interest traverses the forwarding loop $x \rightarrow u \rightarrow v \rightarrow z \rightarrow x$, the Interests from $a$, $b$, and $c$ are aggregated along the loop, which results in a forwarding deadlock. The Interests stored in the PITs of routers $x$, $v$, and $z$ cannot be satisfied with a data packet or a NACK, even though there is a viable path to the source $S$. The Interests must wait in the PITs until their lifetimes expire.

This deadlock  problem is amplified when  Interest pipelining is used, because pipelined Interests from each consumer suffer the same fate until the routing protocol  corrects the forwarding inconsistencies that caused the deadlock.
In our example, the forwarding deadlock was caused by the failure of the interface from $x$ to $k$ and from $v$ to $m$. However, the same forwarding deadlock  could occur if the two interfaces simply became yellow, given the NDN forwarding rules.

The deadlock problem  in NDN  stems from the interaction between the method used by  routers to detect Interest looping, the aggregation of Interests requesting the same content,   and the way in which interfaces are ranked at each router.  
We have shown \cite{ifip2015} that attempting to detect looped Interests using names and nonces  fails to address the fact that  Interests can be  aggregated along forwarding loops, which may result in Interests waiting   for responses that never come and may prevent multicast trees for data to be created. The cause of possible deadlocks in NDN is 
the lack of  ordering among the FIB entries for the same destinations at different routers.

\subsection{Loop-Free and Deadlock-Free Unicast and Multicast  Forwarding in ADN}

We have proven  that  no forwarding loops can occur as long as a router accepts an Interest or datagram only if it carries a distance that is larger than the distance the router has to the intended destination through any of its neighbors  \cite{ifip2015, ancs2015}.
IFR, IAR, and INR implement this sufficient condition for  forwarding and aggregation of Interests. 

Fig.~\ref{adn-loop-fig} illustrates how ADN forwards Interests without forwarding loops or deadlocks using the same example of Fig.~\ref{ndn-loop-fig}. The numbers next to each node in Fig.~\ref{adn-loop-fig}  indicate the distance listed in the FIB of the router for address prefix $n^*$, which is the best match for  the address of the anchor for multicast group $g$.

\begin{figure}
\includegraphics[height=1.5in, width=3in]{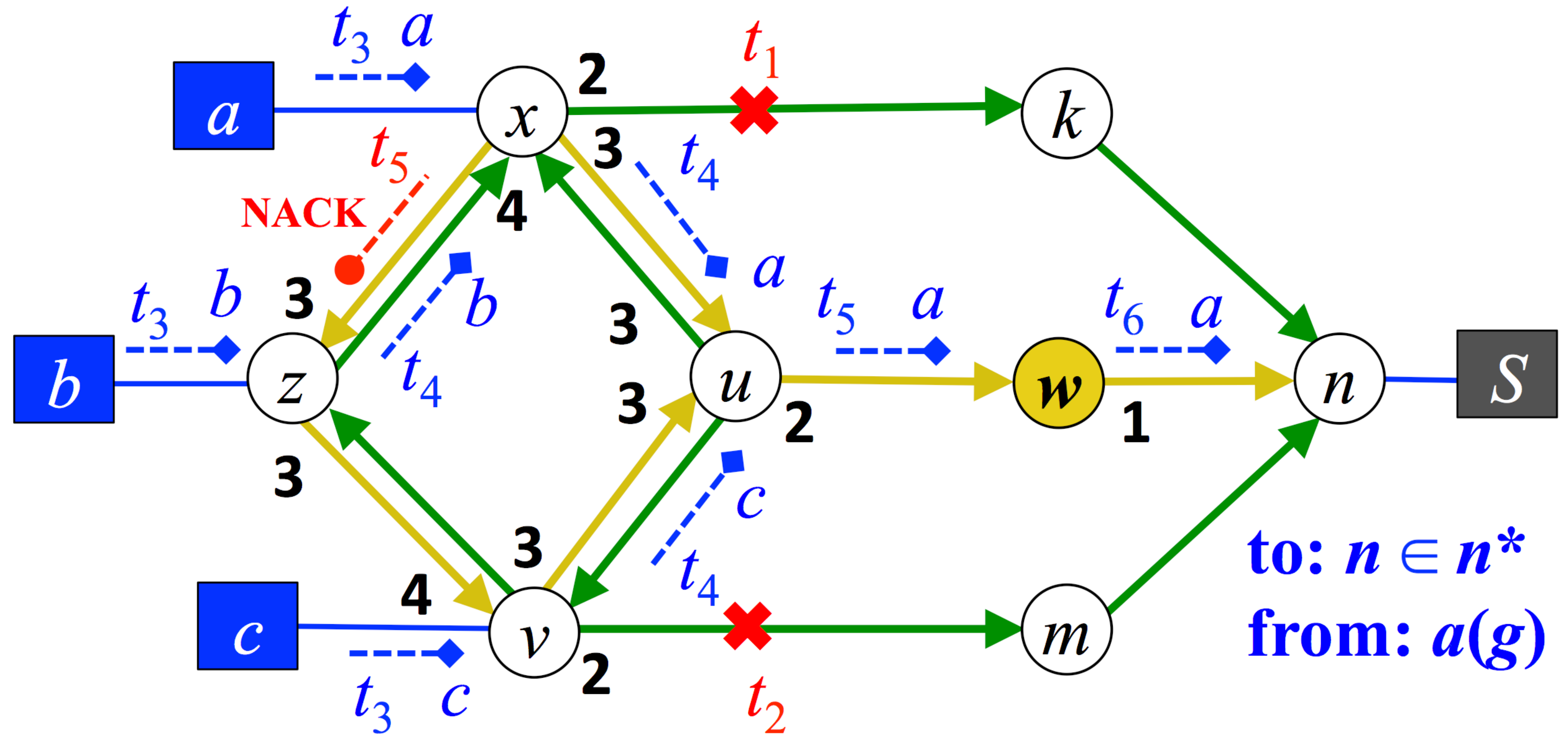}
\caption{No forwarding deadlocks or loops can occur  in ADN}
  \label{adn-loop-fig}
\end{figure}

To have a level playing field in our comparison between ADN and NDN, we assume as in  \cite{icn2016} that routers use a name-based-content routing protocol  similar to the one used in NDN  to populate  their FIBs with the next hops and distances to the addresses of anchors of name prefixes, as well as to populate their DIBs with the mappings of name prefixes to anchor addresses. The number of entries in a DIB is the same as the number of FIB entries in NDN; however, only an ingress routers looks up its DIB before it forwards an Interest; relays simply use their FIBs and SATs. 
We also assume that  routers use  interface rankings similar to those used in NDN to decide how to forward Interests.

Routers $x$, $z$, and $v$ consult their DIBs to bind  the name of the CO being requested to the anchor address $n$ that should receive the Interests. Each CO name $o_i$ is part of a name prefix $g$ for which a best match $g^*$ exists in the DIB.
The multicast group address $a(g)$ is obtained from the group name $g$  itself. 

Routers $x$, $z$ and $v$ forward the Interests they receive from local consumers along the best next hop they find to $n$ at time $t_4$ as shown in Fig.~\ref{adn-loop-fig}. The first Interests  forwarded by the routers for the multicast flow in the example state: 
\\ 
$~~~~~~I[o_i, s^I(x) = g(a), a^I(x) = n, h^I(x)= 3 , f^I(x) \geq 1]$ 
\\
$~~~~~~I[o_i, s^I(z) = g(a), a^I(z) = n, h^I(z)= 3 , f^I(z)  \geq 1]$ 
\\
$~~~~~~I[o_i, s^I(v) = g(a), a^I(v) = n, h^I(v)= 3 , f^I(v) \geq 1]$ 

\vspace{0.05in}
Following INR, router $x$ must send a NACK to the Interest it receives from router $z$ at time $t_5$, because   $h^I(z) = 3 \not> 3 = h_x(n, u)$.
Router $u$ forwards the Interest it receives from $x$ according to IFR at time $t_5$, because because   $h^I(x)  = 3 > 2 = h_u(n, w)$.
Following IAR, router $u$ aggregates the Interest it receives from $v$ because it has created a SAT entry for the multicast address $a(g)$ with $f_u(a(g))= 1$,
$h^I(v)  = 3 > 2 = h_u(n, w)$, and $f^I(v) = f_u(a(g))$. As a result, a single  Interest propagates to $n$ with each router forwarding the Interest creating a SAT entry for $a(g)$. 

A data packet is sent back along the paths stated in the SATs of routers $n$, $w$, $u$, $x$, and $v$ for $a(g)$. Router $z$ propagates a NACK to consumer $b$, which must retransmit its Interest. Eventually, the FIB entries for $n^*$ are corrected 
to state $h_z(n, x) = 4$, which allows $z$ to forward  Interests from $b$ regarding group $g$ through router $x$ or $v$.  Even though consumer $b$ is forced temporarily to retransmit its Interests, each Interest is guaranteed to receive a data packet or a NACK.

The following theorems show that ADN enforces  loop-free forwarding of Interests and Interest aggregation without the possibility of forwarding deadlocks.
The results  are independent of whether the network is static or dynamic, the 
use of in-network caching, or the retransmission and congestion-control strategy used. 

\begin{theorem}
\label{theo1}
Unicast Interests  cannot be forwarded along loops in a network in which ADN is used.
\end{theorem}

\begin{proof}
Consider a network that uses ADN.  Assume for the sake of contradiction that  
routers  in a forwarding loop  $L$ of $h$ hops  $\{ v_1 , $ $v_2 , ..., $ $v_h , v_1 \}$  forward unicast 
Interests for CO $o$ along $L$, with no router  in $L$  detecting the incorrect forwarding of any of the Interests sent over the loop.

Let $a$ be the anchor selected for the Interests for CO $o$ forwarded over $L$, and
let $h^I_L (v_k)$ denote the value of $h^I (v_k)$ when node  $v_k$   forwards the Interest for CO $o$ to
router $v_{k+1}$ for $1 \leq k \leq h - 1$. Similarly, $h^I_L (v_h)$
denotes  the value of $h^I (v_h)$ when when node $v_h$ forwards the Interest for CO $o$ to router  $v_1 \in L$.

Given that  $L$ is formed  by assumption, router $v_k \in L$ must forward 
$I[o, s^I(v_k), a^I(v_k) =a, h^I_L(v_k), f^I(v_k)= 0]$ to router $v_{k+1} \in L$ for $1 \leq k \leq h - 1$, and router $v_h \in L$ must forward Interest 
$I[o, s^I(v_h), a, h^I_L(v_h), f^I(v_h)= 0]$ to router  $v_{1} \in L$. 

According to IFR, if router $v_k$  ($1 < k \leq h$)
forwards 
$I[o, s^I(v_k), $ $ a, h^I_L(v_k), $ $f^I(v_k)= 0]$ to router $v_{k+1}$ as a result of receiving  
$I[o, s^I(v_{k - 1}),$ $a, $ $h^I_L(v_{k - 1}), f^I(v_{k - 1})= 0]$ from router $v_{k-1}$, then  it must be true that  
$h^I_L(v_{k - 1})  > h_{v_k}(a, v_{k+1} ) = h^I_L(v_{k})  $.
Similarly, if router $v_1$ 
forwards Interest 
$I[o, s^I(v_1), $ $ a, h^I_L(v_1), $ $f^I(v_1)= 0]$ to router  $v_{2}$ as a result of receiving Interest 
$I[o, s^I(v_h), $ $ a, h^I_L(v_h), $ $f^I(v_h)= 0]$ from router $v_{h}$, then  it must be true that  
$h^I_L(v_{h})  > h_{v_1}(a, v_{2} ) = h^I_L(v_{1})  $.

However, the above results  constitute a contradiction, because they require that  
$h^I_L(v_{k})  > h^I_L (v_{k}) $ for $1 \leq k \leq h$. Therefore, the theorem is true.  \end{proof}

\begin{theorem}
\label{theo2}
Multicast  Interests  cannot be forwarded or aggregated along loops in a network in which ADN is used.
\end{theorem}

\begin{proof}
We consider a network that uses ADN and assume  that  
routers  in a forwarding loop  $L$ of $h$ hops  $\{ v_1 , $ $v_2 , ..., $ $v_h , v_1 \}$  forward or aggregate multicast 
Interests for group $g$ along $L$, with no router  in $L$  detecting the incorrect forwarding of any of the Interests sent over the loop.

As in the proof of Theorem 1, $h^I_L (v_k)$ denotes the value of $h^I (v_k)$ when node  $v_k$   forwards the Interest for CO $o$ to
router $v_{k+1}$ for $1 \leq k \leq h - 1$ and  $h^I_L (v_h)$
denotes  the value of $h^I (v_h)$ when when node $v_h$ forwards the Interest for CO $o$ to router  $v_1 \in L$, and $a$ denotes the  anchor selected for the Interests for group $g$ forwarded over $L$.

Because no node in $L$ detects the incorrect forwarding  of an Interest, each 
router  in $L$ must aggregate the Interest it receives from the previous hop in $L$ or it must forward the Interest as a result of the Interest it receives from the previous hop in $L$. 

Consider the case in which a router aggregates the Interest it receives from the previous hop along $L$.
In this case IAR must be satisfied at the router.
According to IAR, if $v_k$  ($1 < k \leq h$) aggregates  Interest
$I[o, s^I(v_{k - 1}),$ $a, $ $h^I_L(v_{k - 1}), f^I(v_{k - 1}) > 0]$ received from router $v_{k-1}$, then  it must be true that  
$h^I_L(v_{k - 1})  > h_{v_k}(a, v_{k+1} ) = h^I_L(v_{k})$.
Similarly, if $v_1$ aggregates Interest 
$I[o, s^I(v_h), $ $ a, h^I_L(v_h), $ $f^I(v_h) > 0]$ received from router $v_{h}$, then  it must be true that  
$h^I_L(v_{h})  > h_{v_1}(a, v_{2} ) = h^I_L(v_{1})  $.

On the ether hand, if a router forwards the Interest it receives from the previous hop along $L$, then IFR must be satisfied.
According to IFR, if router $v_k$  ($1 < k \leq h$)
forwards 
$I[o, s^I(v_k), $ $ a, h^I_L(v_k), $ $f^I(v_k) > 0]$ to router $v_{k+1}$ as a result of receiving  
$I[o, s^I(v_{k - 1}),$ $a, $ $h^I_L(v_{k - 1}), f^I(v_{k - 1}) > 0]$ from router $v_{k-1}$, then  it must be true that  
$h^I_L(v_{k - 1})  > h_{v_k}(a, v_{k+1} ) = h^I_L(v_{k})  $.
Similarly, if router $v_1$ 
forwards Interest 
$I[o, s^I(v_1), $ $ a, h^I_L(v_1), $ $f^I(v_1) > 0]$ to router  $v_{2}$ as a result of receiving Interest 
$I[o, s^I(v_h), $ $ a, h^I_L(v_h), $ $f^I(v_h) > 0]$ from router $v_{h}$, then  it must be true that  
$h^I_L(v_{h})  > h_{v_1}(a, v_{2} ) = h^I_L(v_{1})  $.

The above argument renders a contradiction, because  it implies that 
$h^I_L(v_{k})  > h^I_L (v_{k}) $ for $1 \leq k \leq h$. Therefore, the theorem is true.  \end{proof}

For forwarding deadlocks to be avoided, either a data packet with the requested content or a NACK must be received by  the consumer who issued an Interest. 
It follows from Theorems 1 and 2 that no Interest can be forwarded or aggregated along a forwarding loop. Furthermore, according to INR, a router sends a NACK towards the source of an Interest if IFR and IAR are not satisfied, i.e., if no valid forwarding state exists. Accordingly, as long as the consumer uses a valid flow state and there is a viable path between a consumer and a producer or caching site storing the CO requested in an Interest, the consumer must receive the CO.

\section{Performance Comparison}
\label{sec-perf}

We compare ADN and NDN using simulations based on the NDNSim simulation tool.  The performance metrics used for comparison 
are  the average number of forwarding entries needed and  the average end-to-end delayincurred.

The network topology consists of 200 nodes distributed uniformly in a 100m $\times$ 100m area and nodes with distance of 12m or less are connected with point-to-point links of delay 15ms. The data rates of the links are set to 1Gbps to eliminate the effects that a sub-optimal implementation of  CCN-GRAM or NDN may have on the results. 

Each node corresponds to a router and all routers have local   producers and consumers of content, which is the worst-case scenario for ADN.
Interests are generated with a Zipf distribution with  parameter $\alpha = 0.7$ and
producers are assumed to publish 1,000,000 different COs.

We considered {\it total Interest  rates per router} of  50, 100, 200, and 500 objects per second corresponding to the sum of Interests from all local users. The increasing values of total request rates can be viewed as higher request rates from a constant user population of local active users per router, or an increasing  population of active users per router. 

We considered  on-path caching and edge caching. For the case of on-path caching, every router on the path traversed by a data packet towards a consumer caches the CO in its  content store. On the other hand, with edge caching, only the router directly connected to the requesting  consumer  caches the resulting CO.

 \subsection{Size of Forwarding Tables}

Figure \ref{sizePitVSSAT} shows the average size and standard deviation of the number of entries in PITs used in NDN and  SAT entries used in ADN as a function of Interest rates. As the figure shows, the size of PITs grows dramatically as the rate of content requests increases, which is expected given that PITs maintain per-Interest forwarding state. By contrast, the size of SATs remains fairly constant with respect to the content request rates. The figure also shows the average number of Interests received from local consumers pending a response.

For small request rates, the average number of entries in a SAT is actually larger than in a PIT. This is a consequence of using long timers (seconds) to delete SAT entries  independently of whether or not the routes they denote  are actually used by Interests or responses to them. By contrast,  a PIT entry is deleted immediately after an Interest is satisfied.  As the content request rates increase, the size of a PIT can be  more than 10  to 20 times the size of a SAT.
Interestingly, on-path caching offers only minor reductions in forwarding state compared to edge caching for both NDN and ADN.

\begin{figure}
\includegraphics[height=2in, width=3.3in]{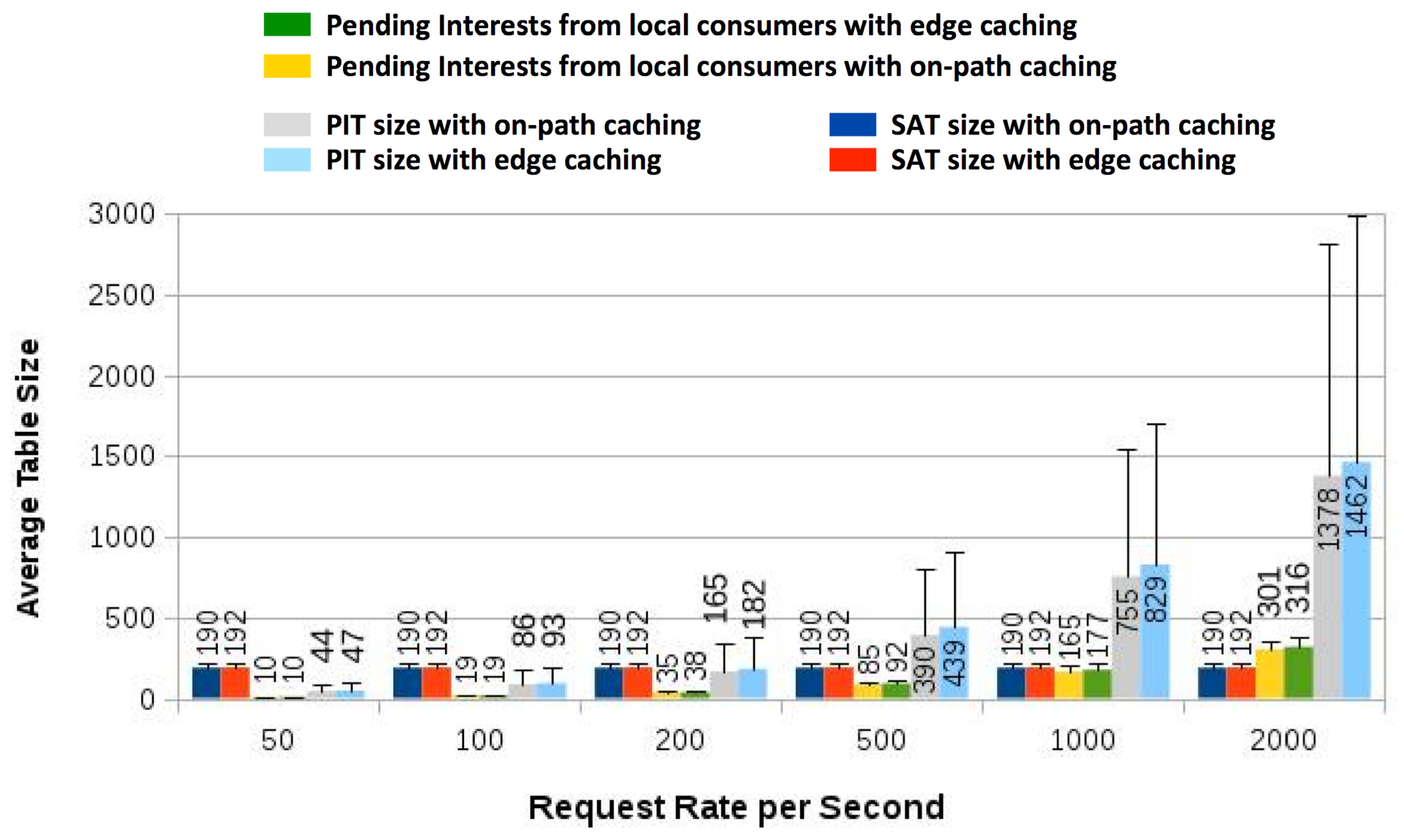}
\caption{Average size of forwarding tables }
   \label{sizePitVSSAT}
\end{figure}

\begin{figure}
\includegraphics[height=2in, width=3.3in]{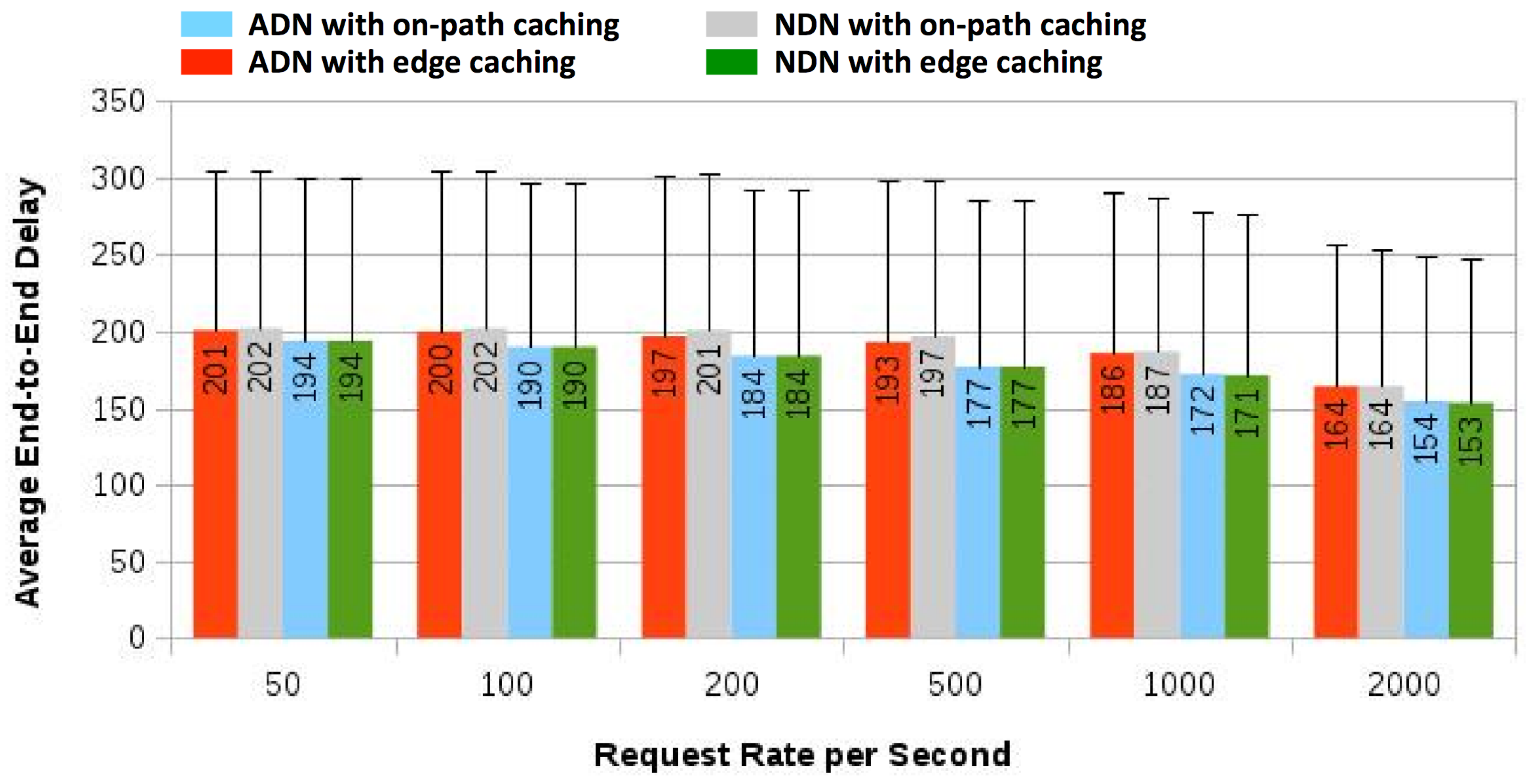}
   \caption{Average end-to-end delays}
   \label{delayADNvsNDN}
\end{figure}

 \subsection{Average Delays}

Figure \ref{delayADNvsNDN} shows the average end-to-end delay for NDN and ADN as a function of   content request rates for on-path caching and edge caching.  As the figure shows, the average delays for NDN and ADN are  essentially  the same for all values of the content request rates.  This should be expected, given that in the experiments the routes in the FIBs for NDN and ADN are static and loop-free.
These results indicate that the number of Interests processed by routers is very similar for NDN and ADN, which justifies the design decision of only using caching as the Interest suppression mechanism in ADN.

\section{Conclusions}
\label{sec-conc}

We argued that 
large  IoT deployments call for the separation of a switching plane in charge of data dissemination and an information plane in charge of discovering resources, content caching, and allowing applications to use expressive names. 
We introduced the Addressable Data Networking (ADN) architecture as an alternative to  the existing IP Internet architecture and NDN or CCNx for large IoT deployments. 
The objective in ADN is to take advantage of the best features of both IP and NDN by implementing a switching plane that works in much the sam sway as the forwarding plane of NDN and CCNx but without the additional overhead incurred in maintaining per-Interest forwarding state. 
We used  limited simulation experiments to highlight the savings in forwarding state attained with ADN compared to NDN, without sacrificing the effectiveness of the network to deliver content.
Effective solutions for content-centric networking in large-scale IoT deployments could be attained through proper modifications of NDN, CCNx, other ICN architectures, or even the IP Internet architecture itself along the lines of what we have proposed as ADN.

\section*{Acknowledgments}

This work was supported in part by the Jack Baskin Chair of Computer Engineering at UC Santa Cruz.

\balance

 \bibliographystyle{ACM-Reference-Format}


\end{document}